\documentclass[submission,copyright,creativecommons]{eptcs}
\providecommand{\event}{SYNT 2014} 
\usepackage{breakurl}             
\usepackage{mathrsfs}
\usepackage{amsmath}
\usepackage{amssymb}

\newtheorem{theorem}{Theorem}[section]
\newtheorem{lemma}[theorem]{Lemma}

\newenvironment{proof}[1][Proof]{\begin{trivlist}
\item[\hskip \labelsep {\bfseries #1}]}{\end{trivlist}}
\newenvironment{definition}[1][Definition]{\begin{trivlist}
\item[\hskip \labelsep {\bfseries #1}]}{\end{trivlist}}

\title{Are There Good Mistakes?\\ A Theoretical Analysis of CEGIS}
\author{Susmit Jha
\institute{Strategic CAD Labs, Intel}
\email{susmit.jha@intel.com}
\and
 Sanjit A. Seshia
\institute{EECS, UC Berkeley}
\email{\quad sseshia@eecs.berkeley.edu}
}
\def\titlerunning{A Theoretical Analysis of CEGIS}
\def\authorrunning{S. Jha \& S.A.Seshia}
\begin{document}
\maketitle

\newcommand{\cegis}{\mathtt{CEGIS}}
\newcommand{\mncegis}{\mathtt{MinCEGIS}}
\newcommand{\hmncegis}{\mathtt{HCEGIS}}
\newcommand{\synth}{\mathtt{IS}}
\newcommand{\lang}{L}
\newcommand{\trace}{\tau}
\newcommand{\nat}{\mathbb{N}}
\newcommand{\range}{range}
\newcommand{\samples}{\mathtt{SMPL}}
\newcommand{\algo}{P}
\newcommand{\template}{\mathtt{TEMPLATE}}
\newcommand{\verifier}{\mathtt{CHECK}}
\newcommand{\mnverifier}{\mathtt{MINCHECK}}
\newcommand{\hmnverifier}{\mathtt{HCHECK}}
\newcommand{\cand}{\mathtt{CAND}}
\newcommand{\ce}{\mathtt{cex}}
\newcommand{\mncemap}{\mathtt{lce}}
\newcommand{\hmncemap}{\mathtt{lce}}
\newcommand{\engine}{T}
\newcommand{\siml}{\mathtt{sim}}

\begin{abstract}
Counterexample-guided inductive synthesis ($\cegis$) is used to synthesize programs from a candidate space of programs.
The technique is guaranteed to terminate and synthesize the correct program if the space of candidate programs is finite.
But the technique may or may not terminate with the correct program if the candidate space of programs is infinite.
In this paper, we perform a theoretical analysis of counterexample-guided inductive synthesis technique.
We investigate whether the set of candidate spaces for which the correct program can be synthesized using $\cegis$ depends on the
counterexamples used in inductive synthesis, that is, whether there are
{\it good mistakes} which would increase the synthesis power.
We investigate whether the use of minimal counterexamples instead of arbitrary counterexamples
expands the set of candidate spaces of programs for which inductive synthesis can successfully synthesize a correct program.
 We consider two kinds of  counterexamples:  minimal counterexamples
 and {\it history bounded} counterexamples. The history bounded
 counterexample used in
 any iteration of $\cegis$ is
 bounded by the examples used in previous iterations of inductive synthesis.
 We examine the relative change in power of inductive
 synthesis in both cases. We show that the synthesis technique using
 minimal counterexamples $\mncegis$ has the
 same synthesis power as $\cegis$
 but the synthesis technique using history bounded counterexamples $\hmncegis$
 has different power than that of $\cegis$, but none dominates the other.
\end{abstract}

\section{Introduction} \label{sec-intro}

Automatic synthesis of programs has been one of the
holy grails of computer science for a long time. It has found many practical
applications such as generating optimal code sequences~\cite{Massalin87,Joshi02},
optimizing performance-critical inner loops, generating
general-purpose peephole optimizers~\cite{Bansal06,bansal08},
automating repetitive programming,
and filling in low-level
details after the higher-level intent has been expressed~\cite{asplos06}.
A traditional view of program synthesis is that of synthesis
from complete specifications. One approach is to give
a specification as a formula in a suitable logic~\cite{Manna80,Manna92fundamentalsof,jha-11}.
Another is to write the specification as a simpler, but
possibly far less efficient program~\cite{Massalin87, asplos06,Joshi02}.
While these
approaches have the advantage of completeness of specification,
such specifications are often unavailable, difficult to
write, or expensive to check against using automated verification
techniques. This has led to proposal of oracle guided synthesis approach~\cite{jha-icse10}
in which the complete specification is not available. All these different variants
of automated synthesis techniques share some common characteristics. They
are iterative inductive synthesis techniques which require some kind of validation
engines to validate candidate programs produced at intermediate iterations, and
these validation engines identify
counterexamples, aka mistakes, which are subsequently used for inductive synthesis in the next iteration.
We collectively refer to such synthesis techniques as
counterexample-guided inductive synthesis, aka $\cegis$.

In this paper, we conduct a theoretical study of $\cegis$  by examining the impact
of using different kinds of validation engines which provide different nature of counterexamples.
$\cegis$ has been successfully used across domains and has been applied to areas such as integer program synthesis and controller design where the candidate set of designs is not finite, and the synthesis technique is not guaranteed to always succeed. This raises an interesting question whether the power of $\cegis$ can be improved by considering validation engines which provide {\it better} counterexamples
than any arbitrary counterexample.

We consider two kinds of counterexamples in this paper.
\begin{itemize}
\item First, we consider {\it minimal} counterexamples instead of arbitrary counterexamples. For any predefined ordering on the examples, we require that the validation engine provide a counterexample which is minimal. This defines an alternative synthesis technique: Minimal Counterexample Guided Inductive Synthesis $\mncegis$ where the validation engine returns {\it minimal} counterexamples.

    This choice of counterexamples is motivated by literature on debugging\footnote{ Practically, this would mean replacing satisfiability solving based verification engines with
    those using Boolean optimization such as maximum satisfiability solving techniques.}. Significant effort has been made on improving validation engines to produce counterexamples which aid debugging by localizing the error. The use of counterexamples in $\cegis$ conceptually is an iterative repair process and hence, it is natural to extend successful error localization and debugging techniques to inductive synthesis. Minimal counterexamples is inspired specifically from  \cite{antonio13,ChenSMV10}.

\item Second, we consider history bounded counterexamples where the counterexample produced by the validation engine must be smaller than a previously seen positive example. This defines another alternative synthesis technique: History Bounded Counterexample Guided Inductive Synthesis $\hmncegis$ where the validation engine returns {\it history bounded} counterexamples.

    This choice of counterexample is also motivated by literature on debugging. In particular, \cite{Groce06,wang-06} use distance of the counterexample from a correct example to help debug
    programs. If the counterexample is very close to a correct example, then the error localization
    would be more accurate. We use a similar notion and force the counterexamples produced by the validation engine to be close to some previously seen correct example.

\end{itemize}

For each of the variants of $\cegis$, we analyze whether it increases
the candidate spaces of programs where a synthesizer terminates with correct program.
We prove the following in the paper.
\begin{enumerate}
\item $\mncegis$ successfully terminates with correct program on a candidate space if and only if
$\cegis$ also successfully terminates with the correct program. So, there is no increase or decrease
in power of synthesis by using minimal counterexamples.
\item $\hmncegis$ can synthesize programs from some program classes where $\cegis$ fails to synthesize
the correct program. But contrariwise, $\hmncegis$ also fails at synthesizing programs from some program
classes where $\cegis$ can successfully synthesize a program. Thus, their synthesis power is not
equivalent, and none dominates the other.
\end{enumerate}

Thus, none of the two counterexamples considered in the paper are strictly {\it good} mistakes.
The history bounded counterexample can enable synthesis in additional classes of programs but
it also leads to loss of some synthesis power.

\section{Motivating Example}

In this section, we present a simple example that illustrates why it is non-intuitive to estimate
the change in power of synthesis when we consider alternative kinds of counterexamples.
Consider synthesizing a program which takes as input a tuple of two integers
$(x,y)$ and outputs $1$ if the tuple lies in a specific rectangle
$R$ (defined by diagonal points $(-1,-1)$ and $(1,1))$ and $0$ otherwise.

The target program is:
$$if\; ( (-1 \leq x \&\& x \leq 1)\&\&(-1\leq y \&\& y\leq 1) ) \; \; {op}=1 \;\; else \;\; op =0$$

The candidate program space is the space of all possible rectangles in $\mathbb{Z} \times \mathbb{Z}$
where $\mathbb{Z}$ denotes the set of integers, that is,

$$if\; ( (\alpha_x \leq x \&\& x \leq \beta_x)\&\&(\alpha_y\leq y \&\& y\leq \beta_y) ) \; \; op=1 \;\; else \;\; op =0$$

where $\alpha_x,\alpha_y,\beta_x,\beta_y$ are the parameters that need to be discovered by the synthesis
engine.

Now, consider a radial ordering of $(x,y)$ which uses $x^2+y^2 $ as the ordering index. If we consider synthesis using minimal counter-examples, it is clear that we can learn the rectangle: starting with an initial candidate program that always outputs $1$ for all $(x,y)$ in $\mathbb{Z} \times \mathbb{Z}$; validation engine producing minimum counterexamples would discover the rectangle boundaries. One possible sequence of minimal counterexamples would be $(0,2),(0,-2),(2,0),(-2,0)$. Since the boundary points form a finite set, $\mncegis$ will terminate with the correct program. But if the counterexamples are arbitrary
as in $\cegis$, it is not obvious whether the rectangle can be still learnt. Our paper proves that
$\cegis$ can also learn such a rectangle.

The question of synthesis power of different techniques using different nature of counterexamples is non-trivial when the space of programs is not finite. Even termination of inductive synthesis technique
is not guaranteed when the candidate space of programs is infinite. Thus, the question of comparing the
relative power of these synthesis techniques is interesting.

\section{Related Work} \label{sec-rel}
Automated synthesis of systems using counterexamples has been widely studied in literature~\cite{solar-thesis08, Srivastava-sigplan10,jha-icse10,henzinger-icalp03,chatterjee-AUAI05}
as discussed in Section~\ref{sec-intro}.
While the applications of $\cegis$ to different domains have been very extensively investigated, theoretical characterization
of the $\cegis$ approach independent of the application domain has received limited attention. To the best of our knowledge, this is the first attempt
at a theoretical investigation into how the nature of counterexamples in $\cegis$ would impact the power of inductive synthesis technique to
synthesize programs.


The inductive generalization used in $\cegis$ is similar to algorithmic learning from examples~\cite{gold67limit,feldman72-ic,Kearns-92,Cornuejols93,feldman-ml91}.
This relation between the two fields has been previously identified in ~\cite{jha-icse10}.
A learning procedure  is provided with strings from a formal language and the task of the learner is to
identify the formal grammar for the language. Learning is an iterative inductive inference
process. In each iteration, the learning procedure is provided a string. The string is either in the language, that is, it is a positive
example, or the string is not in the language, that is, it is negative example. Based on the examples, the learning procedure proposes
a formal grammar in each iteration. The learning procedure is said to be able to learn a formal language if the learner converges to the
correct grammar of the formal language after a finite number of iterations.
The algorithmic learning techniques can be classified across the following three dimensions:
\begin{enumerate}
\item Nature of examples: Examples could be restricted to only positive examples, or it could include negative examples too.
\item Memory of learner: The memory of the learner is allowed to grow infinitely or it could be bounded to a finite size.
\item Communication of examples to learner: The examples could be provided to the learner arbitrarily or as responses to specific kind of queries from the learner such as
membership or subset queries.
\end{enumerate}
We discuss the known theoretical results for algorithmic learning across these dimensions and
identify how the results presented in this paper extend these existing results.

Gold~\cite{gold67limit} considered the problem of learning formal languages from
examples. Similar inductive generalization techniques have been studied elsewhere in literature as well~\cite{JantkeB81,Wiehagen90,blum75,angluin1980inductive}.
The examples
are provided to learner as an infinite stream.
The learner is assumed to have unbounded memory and can store all the examples.
This model is unrealistic in a practical setting but provides useful theoretical understanding
of inductive generalization.
Gold defined a class of languages to be {\it identifiable in the limit} if there is a learning procedure which identifies the grammar
of the target language from the class of languages using a stream of input strings. The languages learnt using only positive
examples were called {\it text learnable} and the languages which require both positive and negative examples were termed
{\it informant learnable}. We examine the known results for both: {\it text learnable} and {\it informant learnable} classes
of languages.
None of the standard classes of formal languages are identifiable in the limit from
text, that is, from only positive examples~\cite{gold67limit}. This includes regular languages, context-free languages and context-sensitive languages. It is also
known that no class of language
with at least one infinite language over the same vocabulary as the rest of the languages in the class, can be learnt purely
from positive examples. We can illustrate this infeasibility of
identifying languages from positive examples with a simple
example.

Consider a vocabulary $V$ and let $V^*$ be all the strings that can be formed using vocabulary $V$. The strings in $V^*$ are $x_1, x_2, \ldots$. Let us consider the set of languages
$$L_1 = V^* - \{x_1\}, L_2 = V^* - \{x_2\}, \ldots$$
Now a simple algorithm to learn languages from positive examples can guess the language to be $L_i$ if $x_i$ is the string with the smallest index not seen so far as a positive example. This algorithm can be used to inductively identify the correct language using
just positive examples. But now, if we add a new language $V^*$ which contains all the strings from vocabulary $V$ to our class of language, that is, $$L_2 = V^*, V^* - \{x_1\}, L_2 = V^* - \{x_2\}, \ldots$$
The above algorithm would fail to identify this class of languages.

In fact, no algorithm using positive examples would be able to inductively identify this class of languages. The key intuition is that if the data is all positive, no finite trace of positive data can distinguish whether the currently guessed language is the target language or is merely a subset of the target language. Now, if we consider the presence of negative counterexamples, the learning or synthesis
algorithm can begin with the first guess as $V^*$. If there are no counterexamples, then $V^*$ is the
correct language. If a counterexample $x_i$ is obtained, then the next guess is $V^* - \{x_i\}$, and
this is definitely the correct language.

A detailed survey of classical results in learning from positive examples
is presented in \cite{Lange08}. The results summarize learning power with different limitations such as
the inputs having certain noise, that is, a string not in the target language might be provided as a positive example with a small probability. Learning using positive as well as negative examples has also been well-studied in literature. A detailed survey is presented in \cite{jain1999systems}
and \cite{lange2000algorithmic}.
In contrast to this line of work, $\cegis$ is a practical inductive generalization which
restricts the memory of the synthesis engine or learner. At any step, the synthesis engine
only has the candidate design and response from the verifier which can be stored in a finite
memory. Further, in contrast to learning from an infinite stream of positive and negative examples, $\cegis$ inductive generalization relies on using counterexamples. The positive and negative examples
used in $\cegis$ are not arbitrary but rather they depend on the
counterexample-generating verifier and the intermediate
candidate programs proposed by the synthesis engine.

Another related line of work is that of techniques using iterative algorithmic
learning with restricted memory of
the learner~\cite{Lange199688,wiehagen76}.
The learner or synthesis engine can use only finite memory but these techniques
rely on availability of an infinite stream of positive examples in addition to negative
examples. While the stream is not explicitly stored due to finite memory constraint, it can
be used for synthesizing intermediate concepts. In contrast, $\cegis$ relies on using
positive examples which are derived from the specification with respect to the counterexamples.
These techniques differ from $\cegis$ in the dimension of how counterexamples
are communicated to the learner or synthesis engine.

Angluin~\cite{angluin88} considered a
similar learning environment as $\cegis$ with respect to the communication of counterexamples
to the learner or synthesis engine. Angluin's learning model consists of a
teacher or oracle which provides responses to queries from the learner. The teacher in the
context of Angluin is analogous to verifier in $\cegis$ and the learner is the synthesis engine. Similar learning models have also being proposed in \cite{shapiro1982algorithmic,Haussler86,Blumer86,hegedus-colt94}.
But they focus on complexity analysis of learning techniques using different kinds
of queries such as membership queries, verification or equivalence queries and subset queries. In contrast, we restrict ourselves to verification queries and investigate the impact of substituting arbitrary counterexample producing verifiers with more powerful
verifiers which produce counterexamples which are minimal or bounded.

Verification techniques have been adapted to provide more meaningful
counterexamples~\cite{Groce06,antonio13,wang-06,ChenSMV10}
for the purpose of aiding design debugging.
The key idea is that these more powerful verification engines
that provide not just any arbitrary counterexamples
but rather a simpler counterexample with respect to some metric
can be used for better debugging. These simpler or minimal
counterexamples
provide the most information to help localize bugs in a faulty design.
If a counterexample trace is close to a correct trace and differs from a correct trace in
a minimal way, then it can be used more effectively to localize the source of bug and fix it.
It is natural to consider extending this use of minimal counterexamples
for debugging to also enable more powerful synthesis.
In this work,
we conduct a theoretical analysis of using these more power verification engines
and using counterexamples produced by these to aid $\cegis$ in synthesis.

\section{Notation} \label{sec-not}

In this section, we define some preliminary notation used in our definition and analysis of CEGIS and MinCEGIS. $\nat$ represents  the set of natural numbers. $\nat_i \subset \nat$ denotes a subset of natural numbers $\nat_i = \{ n | n < i\}$. $\min(S)$ denotes the minimal element in the
set $S$. The union of the sets is denoted by $\cup$ and the intersection of the sets
is denoted by $\cap$.

A sequence $\sigma$ is a mapping from $\nat$ to $\nat \cup \{\bot\}$.
We denote a prefix of length $k$ of a sequence by $\sigma[k]$.
So, $\sigma[k]$ of length $k$ is a mapping from $\nat_k$ to $\nat \cup \{\bot\}$.
$\sigma[0]$ is an empty sequence also denoted by $\sigma_0$. We denote the natural numbers in the range of $\sigma[i]$ by $\samples$, that is, $\samples(\sigma_i) = \range(\sigma_i) - \{\bot\}$.
The set of sequences is denoted by $\Sigma$.

We extend natural numbers to pairs. Let $\langle n_1, n_2 \rangle$ be any
bijective computable function from $\nat \times \nat \rightarrow \nat$
which is monotonically increasing in both of its arguments. Similarly, pairs
can be extended to $n$-tuples. Assuming existence of such a bijective mapping,
tuples can also be used in place of natural numbers as elements of a language.
A language in this case would be a subset of such tuples.

We also use standard definitions from computability theory~\cite{rogers-book87}. A set $\lang_i$ of natural numbers is called computable or recursive  language if there is an program, that is, a computable, total function $\algo_i$ such that for any natural number $n$, $\algo_i(n) = 1$ if $n \in \lang_i$ and $\algo_i(n) = 0$ if $n \not \in \lang_i$. We denote the complement of language $\lang_i$ by $\overline \lang_i$. We denote the union of two languages
$L_i$ and $L_j$ by $L_i \cup L_j$, and the intersection of two languages $L_i$ and $L_j$ by
$L_i \cap L_j$. Also for convenience, we use $\lang(\algo_i)$ to denote $\lang_i$ using the one to one mapping between languages and programs that identify them. Thus, we distinguish only between semantically different
programs and not the syntactically different programs which identify the same language.

The languages are sets of natural numbers. The natural numbers
correspond to indexed
elements of the language or valid input-output traces of the program.
Without loss of generality, the natural ordering
of natural numbers is used as an ordering of elements in the set.
In practice, this will correspond to some user-provided
ordering on the elements of the language.
For example, for a program manipulating strings, we can choose alphabetical ordering
and for program operating on numerical tuples, we can choose lexicographical ordering.
We define a minimum operator $\min(L)$ which uses this natural ordering to report the minimum
element in the language $L$. If the ordering is not total, $\min(L)$ denotes one of the
minimal elements in the language $L$ with respect to the given partial ordering.

Given a sequence $\mathscr{\lang}$ of non-empty languages $\lang_0, \lang_1, \lang_2, \ldots$, $\mathscr{\lang}$ is said to be an indexed family of languages if and only if there exists a recursive function $\template$ such that $\template(i,n) = \algo_i(n)$. We denote the corresponding set of programs $\algo_0, \algo_1, \algo_2, \ldots$ by $\mathscr{\algo}$. For brevity, we refer to $\template(i,n)$ also as $\algo_i(n)$. Intuitively,
$\template$ defines the encoding of candidate program space similar
to sketches in $\cite{asplos06}$ and the component interconnection encoding in
$\cite{jha-icse10}$. The index $i$ is used to index into this encoding to select a particular
program $\algo_i$. $\algo_i(n)$ denotes the output of the program on input $n$.

Inductive synthesis consists of synthesis engines $\engine$ each of which identify the correct program $\algo_i$ using a set of examples from the target language $\lang_i$ from a given indexed family of languages $\mathscr{\lang}$. So, the overall synthesis problem is as follows.
Let $\mathscr{\algo}$ be the class of candidate programs corresponding
 to indexed family of languages $\mathscr{\lang}$.
No, given some target language $\lang_i$ from $\mathscr{\lang}$,
the synthesis engine receives a set of examples
$\{n_1, n_2, \ldots \}$. The synthesis task is to identify $\algo_i$ corresponding to
$\lang_i$ from the candidate programs $\mathscr{\algo}$. $\cegis$ is a particular
kind of inductive synthesis techniques in which examples are obtained using
counterexamples produced through iterative validation of
inductively produced intermediate conjecture programs.  We use the notations
developed in this section to formally define concepts useful for theoretical analysis
of $\cegis$ in the next section.

\section{Definitions}

In this section, we present some definitions. Trace is a sequence of examples from the target language $\lang$.
The formal definition of trace is as follows:

\begin{definition}
 { Trace} $\trace$: A trace $\trace$ for a language $\lang$ is a sequence with
 $\samples(\trace) = \lang$. $\trace[i]$ denotes the prefix of the trace $\trace$ of length $i$.
 $\trace(i)$ denotes the $i$-th element of the trace.
\end{definition}

Counterexample guided inductive synthesis ($\cegis$) techniques employ a verifier to provide counterexamples. So, we define verifiers for a language formally below and then, give a formal definition of a $\cegis$ engine denoted by $\engine_{\cegis}$. Intuitively, the verifier returns a counterexample if the languages are different and returns $\bot$ if they are equivalent. We use one way difference instead of the symmetric difference between sets for ease of presentation.

\begin{definition}
A {verifier} $\verifier_\lang$ for a language $\lang$ is a non-deterministic mapping from $\mathscr{\lang}$ to $\nat \cup \bot$ \\such that $\verifier_\lang(\lang_i) = \bot$ if and only if $\lang_i \subseteq \lang$,
and $\verifier_\lang(\lang_i)  \in \lang_i \cap \overline \lang$ otherwise.
\end{definition}

\begin{definition}
A $\cegis$ engine $\engine_{\cegis} : \sigma \times \sigma \rightarrow \mathscr{\algo}$ is defined recursively below.\\
$\engine_{\cegis} (\trace[n], \ce[n]) = F( \engine_{\cegis} (\trace[n-1], \ce[n-1]),  \trace(n), \ce(n) ) $\\
where $F$ is a recursive function $\mathscr{\algo} \times \nat \times \nat \rightarrow \mathscr{\algo}$ that characterizes the engine and how it eliminates counterexamples, $\trace[n]$ is a trace for language $\lang$ and
$\ce$ is a counterexample sequence such that\\
$\ce(i) = \verifier_\lang(\lang(\engine_{\cegis}(\trace[i-1], \ce[i-1]))) $.\\
$\engine_{\cegis} (\sigma_0, \sigma_0)$ is a predefined constant representing an initial guess $\algo_0$ of the program, which for example,
could be program corresponding to the universal language $\nat$.
\end{definition}

Intuitively, $\cegis$ is provided with a trace along with a counterexample
trace formed by counterexamples to the latest conjectured languages.
Thus, $\cegis$ receives two inputs.
The counterexample is generated through a subset query.

\begin{definition}
We say that $\engine_{\cegis}$ converges to $\algo_i$ if and only if for all, but finitely many prefixes $\trace[n]$ of $\trace$, $\engine_{\cegis}(\trace[n], \ce[n]) = \algo_i$. We denote this by $\engine_{\cegis}(\trace, \ce) \rightarrow \algo_i$.
In other words, $\engine_{\cegis}(\trace, \ce) \rightarrow \algo_i$ if and only if
there exists $k$ such that for all $n\geq k$, $\engine_{\cegis}(\trace[n], \ce[n]) = \algo_i$.
\end{definition}

\begin{definition}
$\engine_{\cegis}$ identifies a language $\lang_i$ if and only if for all traces $\trace$ of the language $\lang_i$
and counterexample sequences $\ce$, $\engine_{\cegis}(\trace, \ce) \rightarrow \algo_i$.
$\engine_{\cegis}$ identifies a language family $\mathscr{\lang}$ if and only if $\engine_{\cegis}$ identifies every $\lang_i \in \mathscr{\lang}$.
\end{definition}

We now define the set of language families that can be identified by the counterexample guided synthesis engines as $\cegis$ formally below.

\begin{definition}
$\cegis = \{ \; \mathcal{L} \; | \; \exists \engine_{\cegis} \; .\; \engine_{\cegis}$ identifies $\mathscr{\lang} \}$
\end{definition}

Now, we consider a variant of counterexample guided inductive synthesis where we use minimal counterexamples instead of arbitrary counterexamples $\mncegis$. We define a minimal counterexample generating verifier before defining $\mncegis$. This requires
an ordering of the elements in the language.

\begin{definition}
A  {verifier} $\mnverifier_\lang$ for a language $\lang$ is a mapping from $\mathscr{\lang}$ to $\nat \cup \bot$ such that\\ $\mnverifier_\lang(\lang_i) = \bot$ if and only if $\lang_i \subseteq \lang$, and
$\mnverifier_\lang(\lang_i) = \min( \overline \lang \cap \lang_i)$ otherwise.
\end{definition}

\begin{definition}
A $\mncegis$ engine $\engine_{\mncegis} : \sigma \times \sigma \rightarrow \mathscr{\algo}$ is defined recursively below.\\
$\engine_{\mncegis} (\trace[n], \ce[n]) = F( \engine_{\mncegis} (\trace[n-1], \ce[n-1]),  \trace(n), \ce(n) ) $\\
where $F$ is a recursive function $\mathscr{\algo} \times \nat \times \nat \rightarrow \mathscr{\algo}$ that characterizes the engine and how it eliminates counterexamples, $\trace[n]$ is a trace for language $\lang$ and
$\ce$ is a counterexample sequence such that\\
$\ce(i) = \mnverifier_\lang(\lang(\engine_{\mncegis}(\trace[i-1], \ce[i-1]))) $.\\
$\engine_{\mncegis} (\sigma_0, \sigma_0)$ is a predefined constant representing an initial guess $\algo_0$ of the program, which for example,
could be program corresponding to the language $\nat$.
\end{definition}

The convergence of the $\mncegis$ synthesis engine to a language and family of languages is defined in similar way as $\cegis$.

\begin{definition}
We say that $\engine_{\mncegis}$ converges to $\algo_i$, that is,
$\engine_{\mncegis}(\trace, \ce) \rightarrow \algo_i$
if and only if
there exists $k$ such that for all
$n\geq k$, $\engine_{\mncegis}(\trace[n], \ce[n]) = \algo_i$.
\end{definition}

\begin{definition}
$\engine_{\mncegis}$ identifies a language $\lang_i$ if and only if for all traces
$\trace$ of the language $\lang_i$
and counterexample sequences $\ce$,
$\engine_{\mncegis}(\trace, \ce) \rightarrow \algo_i$.
$\engine_{\mncegis}$ identifies a language family $\mathscr{\lang}$ if and only if $\engine_{\mncegis}$ identifies every $\lang_i \in \mathscr{\lang}$.
\end{definition}

\begin{definition}
$\mncegis = \{ \; \mathcal{L} \; | \; \exists \engine_{\mncegis} \; .\; \engine_{\mncegis}$ identifies $\mathscr{\lang} \}$
\end{definition}

Next, we consider another variant of counterexample guided inductive synthesis $\hmncegis$ where we use {\it history bounded} counterexamples instead of arbitrary counterexamples. We define a history bounded counterexample generating verifier before defining $\hmncegis$. Unlike the previous
cases, the verifier for generating history bounded counterexample is also provided
with the trace seen so far by the synthesis engine. The verifier generates a counterexample
smaller than the largest element in the trace. If there is no counterexample smaller than the
largest element in the trace, then the verifier does not return any counterexample.
From the definition below, it is clear that we need to only order elements in the language
and do not need to define an ordering of $\bot$ with respect to the language
elements since the comparison is done between an element in non-empty $\lang \cap \lang_i$
and elements $\trace[j]$ in the trace.

\begin{definition}
A  {verifier} $\hmnverifier_\lang$ for a language $\lang$ is a mapping from $\mathscr{\lang} \times \sigma $ to $\nat \cup \bot$ such that\\ $\hmnverifier_\lang(\lang_i,\trace[n]) = m$ where $m \in \overline \lang \cap \lang_i \wedge m < \trace(j)$ for some $j \leq n$, and $\hmnverifier_\lang(\lang_i,\trace[n]) = \bot$ otherwise.
\end{definition}

\begin{definition}
A $\hmncegis$ engine $\engine_{\hmncegis} : \sigma \times \sigma \rightarrow \mathscr{\algo}$ is defined recursively below.
$\engine_{\hmncegis} (\trace[n], \ce[n]) = F( \engine_{\hmncegis} (\trace[n-1], \ce[n-1]),  \trace(n), \ce(n) ) $\\
where $F$ is a recursive function $\mathscr{\algo} \times \nat \times \nat \rightarrow \mathscr{\algo}$ that characterizes the engine and how it eliminates counterexamples, $\trace[n]$ is a trace for language $\lang$ and
$\ce$ is a counterexample sequence such that\\
$\ce(i) = \hmnverifier_\lang(\lang(\engine_{\mncegis}(\trace[i-1], \ce[i-1])), \trace[i-1]) $.\\
$\engine_{\hmncegis} (\sigma_0, \sigma_0)$ is a predefined constant representing an initial guess $\algo_0$ of the program, which for example,
could be program corresponding to the language $\nat$.
\end{definition}

The convergence of the $\hmncegis$ synthesis engine to a language and family of languages is defined in similar way as $\cegis$ and $\mncegis$.

\begin{definition}
We say that $\engine_{\hmncegis}$ converges to $\algo_i$, that is,
$\engine_{\hmncegis}(\trace, \ce) \rightarrow \algo_i$
if and only if
there exists $k$ such that for all
$n\geq k$, $\engine_{\hmncegis}(\trace[n], \ce[n]) = \algo_i$.
\end{definition}

\begin{definition}
$\engine_{\hmncegis}$ identifies a language $\lang_i$ if and only if for all traces
$\trace$ of the language $\lang_i$
and counterexample sequences $\ce$,
$\engine_{\hmncegis}(\trace, \ce) \rightarrow \algo_i$.
$\engine_{\hmncegis}$ identifies a language family $\mathscr{\lang}$ if and only if $\engine_{\hmncegis}$ identifies every $\lang_i \in \mathscr{\lang}$.
\end{definition}

\begin{definition}
$\hmncegis = \{ \; \mathcal{L} \; | \; \exists \engine_{\hmncegis} \; .\; \engine_{\hmncegis}$ identifies $\mathscr{\lang} \}$
\end{definition}

\section{Main Result}
In this section, we present the main results of the paper. We first compare $\mncegis$ and
$\cegis$ in the first part of the section followed by $\hmncegis$ and $\cegis$ in the second part of the
section. Since the focus of our work is to analyze the impact of change in the power of counterexample
providing verification engine, we fix the inductive generalization function $F$ that eliminates counterexamples.
So, we vary the counterexample generating verifier $\verifier_\lang$, $\mnverifier_\lang$ and
$\hmnverifier_\lang$ but $F$ is constant in our  definitions of $\cegis$, $\mncegis$ and $\hmncegis$. In the
rest of the section, we present the two central results of this paper:
\begin{enumerate}
\item $\mncegis = \cegis$
\item $\hmncegis \not = \cegis$
\end{enumerate}

\subsection{Synthesis Using Minimal Counterexamples}
We investigate whether $\mncegis = \cegis$  and prove that it is in fact true. So, replacing a verification engine
which returns arbitrary counterexamples with a verification engine which returns minimal counterexamples does not increase the power of
inductive synthesis system.
The main result regarding this non-intuitive fact that there is no change in the
power of synthesis technique by using
minimal counterexamples is summarized in Theorem~\ref{thm-min}.

\begin{theorem} \label{thm-min}
The power of synthesis techniques using arbitrary counterexamples and those using minimal counterexamples are equivalent, that is, $\mncegis = \cegis$.
\end{theorem}

\begin{proof}
$\cegis \subseteq \mncegis$ trivially. $\mnverifier_\lang$ is a special case of $\verifier_\lang$ and minimal counterexample
reported by $\mnverifier_\lang$ can be treated as arbitrary counterexample to simulate $\cegis$ using $\mncegis$ . Intuitively,
using minimal counterexample is not worse than using arbitrary counterexamples. \\
\\
The more interesting case to prove is $\mncegis \subseteq \cegis$. For a language $\lang$, let $\mncegis$ converge to
$\algo$ on  trace $\trace$. We show that $\engine_{\cegis}$ can simulate $\engine_{\mncegis}$ and also converge to $\algo$
on trace $\trace$.\\
\\
The proof idea is to simulate $\engine_{\mncegis}$ in two phases. In one phase, $\engine_{\cegis}$ finds the minimal
counterexample for a candidate language $L_j$ by iteratively calling $\verifier_\lang$  on $L_j \cap \{i\}$ where
$i = 0,1,2,3\ldots$. The minimum $i$ for which $\verifier_\lang$ returns a counterexample for $L_j \cap \{i\}$ is the
minimum counterexample. In the second phase, $\engine_{\cegis}$ consumes the next elements from the trace.
While searching for minimum counterexample, $\engine_{\cegis}$ needs to store the backlog of the traces as well
as cache the minimum counterexample for candidate languages.\\
\\
We now present the formal description of the proof.
For this simulation, we use some auxiliary variables maintained by $\engine_{\cegis}$ which store some finite information required
for simulating $\engine_{\mncegis}$. The key idea is for $\engine_{\cegis}$ to iteratively guess the minimal counterexample
in multiple micro-steps
and then use that to simulate one step of $\engine_{\mncegis}$. But simulating each step of $\engine_{\mncegis}$ takes finite number
of micro-steps for $\engine_{\cegis}$ and uses finite storage.\\
\\
The first auxiliary component for this simulation is a minimal counterexample map
$$\mncemap: \mathscr{\algo} \rightarrow \nat \cup \{ \top \} \cup \{ \bot \}$$
Intuitively, this maps a candidate program $\algo_i$ (language $\lang_i$) to  minimal counterexample as known to $\engine_{\cegis}$ so far
in simulating $\engine_{\mncegis}$. If minimal counterexample is not known for a given program, $\mncemap$ maps the program to $\top$.
If there is no counterexample to a given program, $\mncemap$ maps the program to $\bot$.
At any given step, only finite number of programs have their minimal counterexamples known, and the rest are mapped to $\top$.

Next, we define a mapping $\engine_{\mncemap}$ from $\mathscr{\algo} \times \sigma \rightarrow \mathscr{\algo}$ which simulates
$\engine_{\mncegis}$ based on the known $\mncemap$ so far, that is,
$$\engine_{\mncemap}(\algo, \trace[n]) = \algo_n \; \mathtt{ where } \; \algo_{i} = F (\algo_{i-1}, \trace(i), \mncemap(\algo_{i-1})) \; \mathtt{for} \; i=1,2,\ldots \; \mathtt{and} \; \algo_0 = \algo$$
if $\mncemap(\algo_{i})$ is defined for $i=1,2,\ldots$ and it is undefined if  $\mncemap(\algo_{i})$ is $\top$ for any $i$. \\
\\
$\engine_{\mncemap}$ simulates $\engine_{\mncegis}$ using the same counterexamples and intermediate candidate programs
for the known history $\mncemap$. If $\mncemap$ is $\top$ for any of the intermediate programs, $\engine_{\mncegis}$
is undefined.
Further, we record the program proposed by $\engine_{\mncegis}$ into the variable $\algo_\siml^m$ and the last
 program which initiated search for minimal counter example in $\algo_{last}$. $\trace_\siml^m$  records
the part of the trace already simulated by $\engine_{\cegis}$ and $\mu$ is the candidate minimal counterexample
while searching for minimal counterexample.
\\
\\
\noindent \textbf{Initialization}:
All the internal auxiliary variables are initialized as follows. $\algo_\siml^0 = \algo_0$ which is the same initialization as
$\engine_{\mncegis}$ being simulated, $\mu = 0$, $\algo_{last} = \algo_0$, and $\trace_\siml^0 = \sigma_0$. $\mncemap$ is initialized
to map all $\algo$ to $\top$ as no minimal counterexamples are known at the beginning.\\
\\
\noindent \textbf{Update}:
We describe the updates made in each iteration $m$. One of the following cases is true in each iteration.\\
\textbf{Case 1}: If $\algo_{last} = \algo_\siml^m$,  that is, we are in lock-step with the $\mncegis$ synthesis algorithm with the same candidate program.\\
\textbf{Case 1.1}: If there is any counterexample for $\algo_\siml^m$ (found using the verifier for $\engine_{\cegis}$), that is, the candidate
program has a counterexample and we need to find the corresponding minimal counterexample.\\
\textbf{Case 1.1.1}: If $\mncemap(\algo_\siml^m)$ is not $\top$, that is, the minimal counterexample for candidate program is already part of
$\mncemap$.\\
Let $\trace_{done}$ be the longest prefix for $\trace_\siml^m \trace(m+1)$ such that $\engine_{\mncemap}(\algo_\siml^m, \trace_{done})$ is defined.
$\trace_{done} \trace_\siml^{m+1} = \trace_\siml^m \trace(m+1)$, $\algo_\siml^{m+1} = \engine_{\mncemap}(\algo_\siml^m, \trace_{done})$, $\algo_{last} = \algo_\siml^{m+1}$\\
We use the minimal counterexample from $\mncemap$ and then advance the simulation $\trace_{done}$ traces ahead if $\engine_{\mncemap}$ can simulate
the trace using minimal counterexamples from $\mncemap$  for all the intermediate candidate programs.\\
\textbf{Case 1.1.2}: If $\mncemap(\algo_\siml^m)$ is $\top$, that is, the minimal counterexample for candidate program is not known.\\
 $\trace_\siml^{m+1} = \trace_\siml^m \trace(m+1)$, $\algo_\siml^{m+1} = \algo_\siml^m \cap \{0\}$\\
We initialize the candidate language $\algo_\siml^{m+1} $ for searching for minimal counterexample to  $\algo_\siml^{m} \cap \{0\}$, that is,
it is either the language consisting only of the minimal element $\{ 0\}$ or is empty. Since our verifier uses a subset query, empty language will return
no counterexamples.  \\
\textbf{Case 1.2}: If there is no counterexample for $\algo_\siml^m$,\\
Let $\trace_{done}$ be the longest prefix for $\trace_\siml^m \trace(m+1)$ such that $\engine_{\mncemap}(\algo_\siml^m, \trace_{done})$ is defined.
$\trace_{done} \trace_\siml^{m+1} = \trace_\siml^m \trace(m+1)$, $\algo_\siml^{m+1} = \engine_{\mncemap}(\algo_\siml^m, \trace_{done})$, $\algo_{last} = \algo_\siml^{m+1}$
and $\mncemap(\algo_\siml^m) = \bot$,\\
The candidate program seen so far is subset of the target language and we consume as much of the trace $\trace_{done}$  as possible
for which $\engine_{\mncemap}$ is defined. \\
\\
\textbf{Case 2}: If $\algo_{last} \not = \algo_\siml^m$, that is, the simulation is trying to find the minimum counterexample as a result of case 1.1.2.\\
\textbf{Case 2.1}: If there is any counterexample $\ce_\siml$ for $\algo_\siml^m$ (found using the verifier for $\engine_{\cegis}$),\\
 Update $\mncemap(\algo_{last}) = \ce_\siml$. Let $\trace_{done}$ be the longest prefix for
$\trace_\siml^m \trace(m+1)$ such that $\engine_{\mncemap}(\algo_{last}, \trace_{done})$ is defined.
$\trace_{done} \trace_\siml^{m+1} = \trace_\siml^m \trace(m+1)$, $\algo_\siml^{m+1} = \engine_{\mncemap}(\algo_{last}, \trace_{done})$, $\algo_{last} = \algo_\siml^{m+1}$, $\mu = 0$\\
If there is a counterexample, since the candidate language was a single element set or empty, and verification engine checks for
containment in the target language, the only element in the language has to be the counterexample. Further, starting from step 1.1.2 and with
possible increments in step 2.2, we stop with the minimal counterexample in this step and add it to the $\mncemap$.\\
\textbf{Case 2.2}: If there is no counterexample for $\algo_\siml^m$, that is, we have not yet found the minimal counterexample.\\
$\mu = \mu +1 $, $\algo_\siml^{m+1} = \algo_{last} \cap \{\mu\}$, and  $\trace_\siml^{m+1} = \trace_\siml^m \trace(m+1)$.\\
We increment $\mu$ and search for whether $\algo_{last} \cap \{\mu\}$ is in the target language. This is either empty or is a language
consisting of a single element $\{\mu\}$.\\
\\
\textbf{Progress}:
Now, we first show progress of the simulation in parsing trace $\trace[m]$. For any $m$, there exists $m' > m$ such that $\trace[m] = \trace_{done}^m \trace_\siml^m$, $\trace[m'] = \trace_{done}^{m'} \trace_\siml^{m'}$ and $\trace_{done}^{m}$ is a proper prefix of
$\trace_{done}^{m'}$. This follows from the observation that Case 2.2 can not be repeated infinitely after Case 1.1.2 since $\algo_{last}$ has at least one counterexample. So, case 2.1 would eventually become true and since $\mncemap$ is extended, $\engine_{\mncegis}$ would be defined for a longer prefix. \\
\\
\textbf{Correctness}:
Let $\engine_{\mncegis}$ converge on $\trace$ after reading prefix $\trace[n]$. From progress, after some $m$, $\trace[n]$ would
be a prefix of $\trace_{done}^m$. Since $\engine_{\mncegis}$  converges after reading $\trace[n]$, $F(\algo_n,\trace(n'),\ce(n')) = \algo_n$
for $n' > n$. Now, $\mncemap$ is not $\top$ for all intermediate programs $\algo_{m''}$ in $\engine_{\mncegis}$ for $m'' \leq m$.
So, $\algo_\siml^{m} = \engine_{\mncemap}(\algo_\siml^0, \trace[m]) = \engine_{\mncemap}(\algo_0, \trace[m]) = \algo_n$ and
for all $m' > m$, $\algo_\siml^{m'} = F(\algo_n,\trace(n'),\ce(n'))  = \algo_n$
So, $\engine_{\cegis}$ also converges to
$\algo_n$, that is, $\mncegis \subseteq \cegis$.\\
\\
Thus, $\mncegis = \cegis$.

\end{proof}

Thus, $\mncegis$ successfully terminates with correct program on a candidate space if and only if
$\cegis$ also successfully terminates with the correct program. So, there is no increase or decrease
in power of synthesis by using minimal counterexamples.

\subsection{Synthesis Using {\it \bf History Bounded} Counterexamples}

We investigate whether $\hmncegis = \cegis$ or not,  and prove that they are not equal.
So, replacing a verification engine
which returns arbitrary counterexamples with a verification engine which returns
counterexamples bounded by history has impact on the power of the synthesis technique.
But this does not strictly increase the power of synthesis. Instead, the
use of history bounded counterexamples does allow programs from
new classes to be synthesized but at the same time, program from some
program classes which were amenable to
$\cegis$ can no longer be synthesized using history bounded counterexamples.
The main result regarding the power of synthesis techniques using
 history bounded counterexamples is summarized in Theorem~\ref{thm-hb}.

\begin{theorem} \label{thm-hb}
The power of synthesis techniques using arbitrary counterexamples and those using
history bounded counterexamples are not equivalent, and none is more powerful than the other.
$\hmncegis \not = \cegis$. In fact,  $\hmncegis \not \subseteq \cegis$ and $\cegis \not \subseteq \hmncegis$.
\end{theorem}

We prove this using the following two lemma. The first lemma \ref{lemma-cegisbetter}
shows that there is a family of
languages from which a program recognizing a language can be synthesized by $\cegis$ but,
this can not be done by $\hmncegis$. The second lemma \ref{lemma-hcegisbetter} shows that
there is another family of languages from which a program recognizing a language can be synthesized by
$\hmncegis$ but not by $\cegis$.

\begin{lemma} \label{lemma-cegisbetter}
There is a family of languages $\mathcal{L}$ such that for the candidate programs $\mathcal{\algo}$ corresponding to
$\mathcal{L}$, $\hmncegis$ can not synthesize a program $\algo$ in  $\mathcal{\algo}$
recognizing some language $\lang$ in $\mathcal{L}$ but
$\cegis$ can synthesize  $\algo$, that is,
 $\cegis \not \subseteq \hmncegis$
\end{lemma}

\begin{proof}
Consider the languages formed by upper bounding the elements by some fixed constant, that is,
$$\lang_i = \{ n | n \in \nat \wedge n \leq i \}$$
Now, consider the family of languages consisting of these, that is,
$\mathcal{L} = \{ \lang_i | i \in \nat \}$.
Given this family $\mathcal{L}$, let the target language $\lang$ (for which we want to synthesize a recognizing program $\algo$) be $\lang_i$.

If we obtain a trace $\trace[j]$ at any point in synthesis
using history bounded counterexamples, then for any intermediate program $\algo_j$ proposed by
$\engine_{\hmncegis}$, $\hmnverifier_\lang$  would always return $\bot$ since all the counterexamples
would be larger than any element in $\trace[j]$. This is the consequence of the chosen languages
in which all counterexamples to the language are larger than any positive example of the language.
So, $\engine_{\hmncegis}$ can not synthesize $\algo$ corresponding to the target language $\lang$.

But we can easily design a synthesis engine $\engine_{\cegis}$
using arbitrary counterexamples that can synthesize
$\algo$ corresponding to the target language $\lang$. The algorithm starts with $L_0$ as its
initial guess. If there is no counterexample, the algorithm next guess is $\lang_1$.
In each iteration $j$,
the algorithm guesses $\lang_{j+1}$ as long as there are no counterexamples. When a counterexample
is returned by $\verifier_\lang$ on the guess $\lang_{j+1}$,
the algorithm stops and reports the previous guess $\lang_j$ as the correct language.

Since the elements in each language $\lang_i$ is bounded by some fixed constant $i$, the above synthesis procedure $\engine_{\cegis}$ is guaranteed to terminate after $i$ iterations when identifying any language $\lang_i \in \mathcal{L}$. Further, $\verifier_\lang$ did not return any counterexample up to iteration $j-1$ and so, $\lang_j \subseteq \lang_i$. And in the next iteration, a counterexample was generated.
So, $\lang_{j+1} \not \subseteq \lang_i$. Since, the languages in $\mathcal{L}$ form a monotonic chain
$\lang_0 \subset \lang_1 \ldots $. So, $\lang_j = \lang_i$. In fact $j=i$ and in the $i$-th iteration,
the language $L_i$ is correctly identified by $\engine_{\cegis}$.

Thus,  $\cegis \not \subseteq \hmncegis$.
\end{proof}

This shows that $\cegis$ can be used to identify programs when $\hmncegis$ will fail. Putting a
restriction on the verifier to only produce counterexamples which are bounded by the
positive examples seen so far does not strictly increase the power of synthesis.

We now show the nonintuitive result that this restriction enables synthesis
of programs which can not be synthesized by $\cegis$. The proof uses a
diaganolization argument similar to the argument used in \cite{Lange08} for showing the
increase in inductive synthesis power when negative examples are introduced in addition
to the positive examples. This argument is presented in Section~\ref{sec-rel}.
Recall that the set of languages considered in that case were
$L_1 = V^* - \{x_1\}, L_2 = V^* - \{x_2\}, \ldots$ and the language $V^*$. The argument
relies on indistinguishability of  $V^* - \{x_1\}$ and  $V^*$ with respect to finite traces
of positive examples.

In the proof below, we similarly construct a language which is not distinguishable
using arbitrary counterexamples and instead, it relies on the verifier keeping a record of the
largest positive example seen so far and restricting counterexamples to those below the largest positive
example. We use the tuple notation introduced in Section~\ref{sec-not} to clearly identify the
diagnolization.

\begin{lemma} \label{lemma-hcegisbetter}
There is a family of languages $\mathcal{L}$ such that for the candidate programs $\mathcal{\algo}$ corresponding to
$\mathcal{L}$, $\cegis$ can not synthesize a program $\algo$ in  $\mathcal{\algo}$
recognizing some language $\lang$ in $\mathcal{L}$ but
$\hmncegis$ can synthesize  $\algo$, that is,
 $\hmncegis \not \subseteq \cegis$
\end{lemma}

\begin{proof}
Consider the following languages $\lang^{01}_i = \{ \langle j, n \rangle | j \in \{0,1\}, n \in \nat \}  $. We now construct a family of languages in $\lang^{01}_i$ which are finite and have atleast one
 element of the form $\langle 1, . \rangle$, that is,
$$\mathcal{L}^{fin} = \{\lang^{01}_i | i \in \nat \wedge |\lang^{01}_i| \text{is finite} \wedge \exists k \; s.t. \; \langle 1, k \rangle \in \lang^{01}_i \}$$
Now consider the languages $L_i$ which are subsets of $\nat$. We consider only
those languages
$L_i$ such that the index $i$ of the language is also the smallest element in the language,
that is, $ i = \min(L_i)$.
We now build a language of pairs as follows:
$L^{diag}_i = \{\langle 0, n \rangle | n \in L_i \}$ if $ i = \min(L_i)$ and undefined, otherwise
We construct a second family of languages using these languages.
$\mathcal{L}^{diag} = \{ \lang^{diag}_i \} $ if $\lang^{diag}_i$ is defined for index $i$.
Now, we consider the following family of languages
$$ \mathcal{L} = \mathcal{L}^{fin} \cup \mathcal{L}^{diag}$$
We show that there is a language $\lang$ in $\mathcal{L}$ such that the program $\algo$ recognizing $\lang$ can not be synthesized by $\cegis$ but $\hmncegis$ can synthesize all programs recognizing
any language in $\mathcal{L}$.

The key intuition is as follows.
If the examples seen by synthesis algorithm are all of the form $\langle 0, . \rangle$, then
any synthesis technique can not differentiate whether the
language belongs to $\mathcal{L}^{fin}$ or $\mathcal{L}^{diag}$. If the language belongs
to  $\mathcal{L}^{fin}$, the synthesis engine would eventually obtain
an example of the form  $\langle 1, . \rangle$ (since each language in $\mathcal{L}^{fin}$
has atleast one element of this kind and these languages are finite). While
the synthesis technique using arbitrary counterexamples can not recover the previous examples,
the techniques with access to the verifier which produces history bounded counterexamples
can recover all the previous examples.

We can easily specify a $\engine_\hmncegis$ which can synthesize programs
that correspond to languages in $\mathcal{L}$. $\engine_\hmncegis$ works as follows.
If all the elements seen so far are of the form $\langle 0, . \rangle$,
then the synthesis algorithm and picks the
minimum $j$ such that $\langle 0, j \rangle$ has been seen as an example by the synthesis engine.
The proposed program is $\algo_j$ corresponding to $\lang_j$.
If the proposed
program is not the correct program, $\hmnverifier$ returns $\langle 0, j_{ce} \rangle$
such that $j_{ce} < j$. This is guaranteed since $\hmnverifier$ returns counterexamples
smaller than the examples seen so far, and we have assumed that $\algo_j$ is not correct.
So, iteratively, the algorithm would discover a language from $\mathcal{L}^{diag}$
eventually.  But if the language is from $\mathcal{L}^{fin}$, then we know that
all languages in $\mathcal{L}^{fin}$ are finite and have at
least one element of the form $\langle 1,.\rangle$.
After $\engine_\hmncegis$  sees $\langle 1,.\rangle$, for every future positive
example $x$, it queries $\hmnverifier$ with the singleton language having
only one element $\{\langle x +2, 0 \rangle \}$. Clearly, $\langle x +2, 0 \rangle $ is not in the language since
it only contains elements of the form $\langle 0, . \rangle$ and $\langle 1, . \rangle$.
But $\hmnverifier$
returns no counterexample for $\{\langle x_{max} +2, 0 \rangle \}$  if $x_{max}$ is the largest positive example seen so far.
At this point, we can recover all positive examples seen previously by enumerating all $x' < x_{max}$
and testing the candidate language $\{ x' \}$ with $\hmnverifier$.
We get a counterexample if and only if $x'$ is not in the target language.
Further, the target language is finite and hence, enumerating members of the language is sufficient
to identify the target language after consuming a finite trace.
Thus, $\engine_\hmncegis$  can synthesize programs corresponding to any language in
$\mathcal{L}$.

We now prove the infeasibility of $\cegis$ for this class of languages. Let us
assume that $\mathcal{L} \in \cegis$. So, there is a synthesis engine $\engine_\cegis$
which can synthesize programs corresponding to languages in $\mathcal{L}$.
Let us consider trace $\trace$ and counterexample sequence $\ce$ such that $\engine_\cegis$
converges in $n$ steps that is $\engine_\cegis(\trace, \ce) \rightarrow \engine_\cegis(\trace[n], \ce[n])$.
Now, $\ce[n]$ is a valid counterexample sequence of any language $L$ such that
$\samples(\trace[n]) \subseteq L \subseteq \nat - \samples(\ce[n])$.
Since $\engine_\cegis$ must
recognize a language from any trace and any arbitrary counterexample sequence, we choose
a trace and counterexample sequence as follows.
Let us consider a trace $\trace'$  of the form $\trace[n] (\langle 1, z_1 \rangle)^\infty$.
The corresponding counterexample trace discovered by $\engine_\cegis$ is $\ce[n]$ followed
by minimal counterexamples, if any, after observing $\langle 1, z_1 \rangle$.
Now, we pick an element
$\langle 0, z_2 \rangle$ such that  $\langle 0, z_2 \rangle \not \in \ce[n]$
and $\langle 0, z_2 \rangle \not \in \trace[n]$.
Since $\ce[n]$ is a valid counterexample sequence of any language $L$ such that
$\samples(\trace[n]) \subseteq L \subseteq \nat - \samples(\ce[n])$, the behavior
of $\engine_\cegis$ is same for $\trace[n] (\langle 1, z_1 \rangle)^\infty$
as it is for $\trace[n] \langle 0, z_2 \rangle (\langle 1, z_1 \rangle)^\infty$.
Thus, $\engine_\cegis$ can not distinguish between the
two languages: $\lang^d = \samples(\trace[n]) \cup \{ \langle 1, z_1 \rangle) \}$
and  $\lang^{d'} = \samples(\trace[n]) \cup \{ \langle 0, z_2 \rangle, \langle 1, z_1 \rangle) \}$.
Intuitively, $\engine_\cegis$ can forget some positive examples seen before observing  $\langle 1, z_1 \rangle$
and there is no way to regenerate these as it can be done with $\engine_\hmncegis$.

Thus, $\hmncegis \not \subseteq \cegis$.

\end{proof}

Hence, $\hmncegis$ can synthesize programs from some program classes where $\cegis$ fails to synthesize
the correct program. But contrariwise, $\hmncegis$ also fails at synthesizing programs from some program
classes where $\cegis$ can successfully synthesize a program. Thus, their synthesis power is not
equivalent, and none dominates the other.

\section{Discussion and Conclusion}
The paper presents formal analysis of the impact of counterexample selection on
what programs can be synthesized, without any restriction on the type of program other than it be from a countable set.
We have shown that the use of minimal counterexamples does not enable synthesizing programs from newer space of candidate programs.
In practice, this means that any domain where $\mncegis$ can be used, use of $\cegis$ would also be possible
since $\mncegis$ successfully terminates with correct program on a candidate space if and only if
$\cegis$ also successfully terminates with the correct program. So, there is no increase or decrease
in the power of synthesis by using minimal counterexamples.
But $\hmncegis$ can synthesize programs from some program classes where $\cegis$ fails to synthesize
the correct program. Contrariwise, $\hmncegis$ also fails at synthesizing programs from some program
classes where $\cegis$ can successfully synthesize a program. Thus, their synthesis power is not
equivalent, and none dominates the other. This paper is a first step towards the theoretical characterization of Counterexample Guided Inductive Synthesis technique: $\cegis$.

Further analysis of $\cegis$ is pertinent given the widespread adoption of $\cegis$ as one of the
standard paradigms for automated synthesis. We envision the following directions in which further work
can be done to better understand the power of $\cegis$ techniques.
\begin{itemize}
\item Speed of convergence: $\mncegis$ and $\cegis$ have equal synthesis power and if one of the techniques successfully identifies a program from a given program class, the other would also be able to successfully synthesize this program. But would both techniques need the same number of counterexamples for successfully synthesizing the program ? If we measure the complexity of automated synthesis using the number of counterexamples needed to synthesize a program, the comparison of the complexity of $\mncegis$ and $\cegis$ is open.

    Similarly, for the program spaces on which both $\hmncegis$ and $\cegis$ terminate, can we compare the number of counterexamples needed by the two techniques to synthesize a program.

\item Newer variants of counterexamples: The two new variants of counterexamples considered in this paper; namely, the minimal counterexamples and the history bounded counterexamples are not the only variants that can be used in $\cegis$.
    The question of whether there are other variants of counterexamples which would enable synthesis in program spaces beyond the power of conventional $\cegis$ is open.

     In particular, consider  another new variant of counterexamples which are minimal counterexamples among all the counterexamples which are larger than the largest positive examples seen so far. This counterexample captures another notion of being {\it close to correct} counterexample, and it would be interesting to investigate whether it increases the power of $\cegis$.
\end{itemize}

In summary, we presented variants of $\cegis$ using different kinds of counterexamples
and compared the power of these variant synthesis techniques with $\cegis$. This is a first step
towards a better theoretical understanding of the synthesis power of $\cegis$ technique.

\bibliographystyle{eptcs}
\bibliography{paper}

\end{document}